\documentclass[11pt,letterpaper]{article}  

\usepackage{amsthm} 
\usepackage{amssymb}  

\newtheorem{theorem}{Theorem}
\newtheorem{corollary}{Corollary}
\newtheorem{lemma}{Lemma}
\newtheorem{remark}{Remark}

\newtheorem{assumption}{Assumption}
\usepackage[margin=1in]{geometry}                                                    
\usepackage{amsmath,amsfonts,amssymb,color}
\usepackage{epsfig}
\usepackage{algorithm}
\usepackage{algorithmic}
\usepackage{multirow}

\usepackage{array}
\usepackage{multirow}
\usepackage{epstopdf}
\usepackage{tikz}
\usepackage{relsize}
\usetikzlibrary{shapes,arrows}
\usepackage{cite}
\usepackage{pmat}
\usepackage{epstopdf}
\usepackage{tikz}
\usetikzlibrary{shapes}
\usepackage{scalerel}
\usepackage[hidelinks]{hyperref}
\hypersetup{
    colorlinks,
    linkcolor={blue!},
    citecolor={blue!}
}
\DeclareMathOperator*{\argmin}{\arg\!\min}

\title{\LARGE \bf Finite-Time Distributed State Estimation over Time-Varying Graphs: Exploiting the Age-of-Information}
\author{Aritra Mitra, John A. Richards, Saurabh Bagchi and Shreyas Sundaram
\thanks{A. Mitra, S. Bagchi, and S. Sundaram are with the School of Electrical and Computer Engineering at Purdue University. J. A.  Richards is with Sandia National Laboratories.   Email: {\tt \{mitra14, sbagchi, sundara2\}@purdue.edu},  {\tt{jaricha@sandia.gov}}. This work was supported in part by NSF CAREER award
1653648, and by a grant from Sandia National Laboratories. Sandia National Laboratories is a multimission laboratory managed and operated by National Technology \& Engineering Solutions of Sandia, LLC, a wholly owned subsidiary of Honeywell International Inc., for the U.S. Department of Energy's National Nuclear Security Administration under contract DE-NA0003525. The views expressed in the article do not necessarily represent the views of the U.S. Department of Energy or the United States Government.}}
\begin{document}
\maketitle
\begin{abstract}
We study the problem of collaboratively estimating the state of a discrete-time LTI process by a network of sensor nodes interacting over a time-varying directed communication graph. Existing approaches to this problem either (i) make restrictive assumptions on the dynamical model, or (ii) make restrictive assumptions on the sequence of communication graphs, or (iii) require multiple consensus iterations between consecutive time-steps of the dynamics, or (iv) require higher-dimensional observers. In this paper, we develop a distributed observer that operates on a single time-scale, is of the same dimension as that of the state, and works under mild assumptions of joint observability of the sensing model, and joint strong-connectivity of the sequence of communication graphs. Our approach is based on the notion of a novel ``freshness-index" that keeps track of the age-of-information being diffused across the network. In particular, such indices enable nodes to reject stale information regarding the state of the system, and in turn, help achieve stability of the estimation error dynamics. Based on the proposed approach, the estimate of each node can be made to converge to the true state exponentially fast, at any desired convergence rate. In fact, we argue that finite-time convergence can also be achieved through a suitable selection of the observer gains. Our proof of convergence is self-contained, and employs simple arguments from linear system theory and graph theory.
\end{abstract}
\section{Introduction}
Given a discrete-time LTI system $\mathbf{x}[k+1]=\mathbf{Ax}[k]$, and a linear measurement model $\mathbf{y}[k]=\mathbf{Cx}[k]$, a classical result in control theory states that one can design an observer that generates an asymptotically correct estimate $\hat{\mathbf{x}}[k]$ of the state $\mathbf{x}[k]$, if and only if the pair $(\mathbf{A},\mathbf{C})$ is detectable. Additionally, if the pair $(\mathbf{A},\mathbf{C})$ is observable, then one can achieve exponential convergence at any desired convergence rate. Over the last couple of decades, significant effort has been directed towards studying the distributed counterpart of the above problem, wherein observations of the process are distributed among a set of sensors modeled as nodes of a communication graph \cite{dist3,ugrinov,
kim,ren,martins,allerton,wang,han,mitraTAC,rego,wang2}. A fundamental question that arises in this context is as follows:  What are the minimal requirements on the measurement structure of the nodes and the underlying communication graph that guarantee the existence of a distributed observer? Here, by a distributed observer, we imply a set of state estimate update and information exchange rules that enable each node to track the entire state asymptotically. The question posed above was answered only recently in \cite{martins,allerton,wang,mitraTAC,han,rego} for static graphs.

The approaches in \cite{martins,allerton,wang,mitraTAC,wang2,han,rego} can be generally classified based on the following attributes. (i) Does the approach require multiple consensus iterations between two consecutive time-steps of the dynamics?\footnote{Such approaches, referred to as two-time-scale approaches, may prove to be computationally prohibitive for real-time applications.} (ii) What is the dimension of the estimator maintained by each node? (iii) Can the convergence rate be controlled? (iv) Is the approach robust to temporal variations in the underlying communication graph? The techniques proposed in \cite{martins,allerton,mitraTAC,wang,han,rego} operate on a single-time-scale, those in \cite{allerton,mitraTAC,han,rego} require observers of dimension no more than that of the state of the system, the ones in \cite{wang,wang2,han} can achieve any desired convergence rate, while the one in \cite{wang2} can account for a fairly general class of time-varying graphs. The main \textbf{contribution} of this paper is the development of a distributed observer that shares each of the above positive attributes. Specifically, we develop a single-time-scale distributed state estimation algorithm in Section \ref{sec:algo} that requires each node to maintain an estimator of dimension equal to that of the state (along with some simple counters), and works under the basic assumptions of joint observability of the observation model, and joint strong-connectivity of the sequence of communication graphs.\footnote{In Remark \ref{rem:assump}, we explain that each of these assumptions can be further relaxed.}

The authors in \cite{cao} point out that even for the basic consensus problem, understanding the difference between static graphs and time-varying graphs is much the same as understanding the difference between the stability of LTI systems and LTV systems, the implication being that the extension is highly non-trivial. Arguably, the stability analysis for the distributed state estimation problem with time-varying graphs is even more challenging, since one has to account for potentially unstable external dynamics, a feature that is missing in the standard consensus problem. Consequently, one can no longer directly leverage convergence properties of products of stochastic matrices. Nevertheless, in Section \ref{sec:mainresult}, we establish using simple arguments from linear system theory and graph theory, that based on our approach, each node can track the true dynamics exponentially fast at any desired convergence rate. Additionally, we show how to design the observer gains so as to achieve convergence in finite time. The closest related work is reported in \cite{wang2}, where the authors study a continuous-time analog of the problem under consideration, and develop a solution that leverages an elegant connection to the problem of distributed linear-equation solving \cite{mou}. In contrast to our technique, the one in \cite{wang2} is inherently a two-time-scale approach, requires each node to maintain and update auxiliary state estimates, and works under the assumption that the communication graph is strongly-connected at every time-instant.

The key idea behind our algorithm is the use of a suitably defined ``freshness-index" that keeps track of the age-of-information being diffused across the network. Loosely speaking, such indices are a measure of the accuracy with which the information received by a node describes the physical process being observed. While the freshness-indices enable a node to reject stale information, the assumption of joint strong-connectivity ensures that fresh information is diffused across the network sufficiently often. These facts taken together help achieve stability of the estimation error process. Finally, we point out that while this is perhaps the first use of the notion of age-of-information (AoI) in a networked control/estimation setting, such a concept has been widely  employed in the study of various queueing-theoretic problems arising in wireless networks \cite{kaul,huang,talak}.\footnote{The notion of age-of-information (AoI) was first introduced in \cite{kaul} as a performance metric to keep track of real-time status updates in a communication system. In the context of a wireless network, it measures the time elapsed since the generation of the packet most recently delivered to the destination. In Section \ref{sec:algo}, we will see how such a concept applies to the present setting.}

\section{Problem Formulation and Background}
We are interested in collaborative state estimation of a discrete-time LTI system of the form:
\begin{equation}
\mathbf{x}[k+1]=\mathbf{A}\mathbf{x}[k],
\label{eqn:system}
\end{equation}
where $k\in\mathbb{N}$ is the discrete-time index, $\mathbf{A}\in\mathbb{R}^{n \times n}$ is the system matrix, and $\mathbf{x}[k]\in\mathbb{R}^{n}$ is the state of the system.\footnote{We use $\mathbb{N}$ and $\mathbb{N}_{+}$ to denote the set of non-negative integers and the set of positive integers, respectively.} A network of sensors, modeled as nodes of a communication graph, obtain partial measurements of the state of the above process as follows:
\begin{equation}
\mathbf{y}_{i}[k]=\mathbf{C}_i\mathbf{x}[k],
\label{eqn:Obsmodel}
\end{equation}
where $\mathbf{y}_{i}[k] \in {\mathbb{R}}^{r_i}$ represents the measurement vector of the $i$-th node at time-step $k$, and $\mathbf{C}_i \in {\mathbb{R}}^{r_i \times n}$ represents the corresponding observation matrix. Let $\mathbf{y}[k]={\begin{bmatrix}\mathbf{y}^T_{1}[k] & \cdots & \mathbf{y}^T_{N}[k]\end{bmatrix}}^T$ and $\mathbf{C}={\begin{bmatrix}\mathbf{C}^T_{1} & \cdots & \mathbf{C}^T_{N}\end{bmatrix}}^T$ represent the collective measurement vector at time-step $k$, and the collective observation matrix, respectively. The goal of each node $i$ in the network is to generate an asymptotically correct estimate $\hat{\mathbf{x}}_i[k]$ of the true dynamics $\mathbf{x}[k]$. It may not be possible for any node $i$ in the network to accomplish such a task in isolation, since the pair $(\mathbf{A},\mathbf{C}_i)$ may not be detectable in general. Throughout the paper, we will only assume that the pair $(\mathbf{A},\mathbf{C})$ is observable. 

As is evident from the above discussion, information exchange among nodes is necessary to solve the problem at hand. At each time-step $k\in\mathbb{N}$, such interactions are modeled by a directed communication graph $\mathcal{G}[k]=(\mathcal{V},\mathcal{E}[k])$, where $\mathcal{V}=\{1,\ldots,N\}$ represents the set of nodes, and $\mathcal{E}[k]$ represents the edge set of $\mathcal{G}[k]$ at time-step $k$. Specifically, if $(i,j)\in\mathcal{E}[k]$, then node $i$ can send information directly to node $j$ at time-step $k$; in such a case, node $i$ will be called a neighbor of node $j$ at time-step $k$. We will use $\mathcal{N}_i[k]$ to represent the set of  all neighbors (excluding node $i$) of node $i$ at time-step $k$. When $\mathcal{G}[k]=\mathcal{G}  \hspace{2mm}\forall k\in\mathbb{N}$, where $\mathcal{G}$ is a static, directed communication graph, the necessary and sufficient condition (on the system and network) to solve the distributed state estimation problem is that each source component of $\mathcal{G}$ be collectively detectable \cite{martins}.\footnote{A source component of a directed graph is a strongly connected component with no incoming edges.} Our goal in this paper is to extend the above result to the scenario where the underlying communication graph is allowed to change over time. To this end, let the union graph over an interval $[k_1,k_2], 0 \leq k_1 < k_2$, indicate a graph with vertex set equal to $\mathcal{V}$, and  edge set equal to the union of the edge sets of the individual graphs appearing over the interval $[k_1,k_2]$. Based on this convention, we will assume that the sequence of communication graphs $\{\mathcal{G}[k]\}_{k=0}^{\infty}$ is ``jointly strongly-connected", in the sense described below.

\begin{assumption}\textbf{(Joint Strong-Connectivity)} There exists ${T}\in\mathbb{N}_{+}$ such that the union graph over every interval of the form $[kT,(k+1)T)$ is strongly-connected, where $k\in\mathbb{N}$.
\label{assump:connectivity}
\end{assumption}

For communication graphs satisfying the above assumption, our \textbf{objective} will be to design a distributed algorithm that ensures $\lim_{k\to\infty}\left\Vert\hat{\mathbf{x}}_i[k]-\mathbf{x}[k]\right\Vert=0, \forall i\in\mathcal{V}$.

To this end, we recall the following result from \cite{mitraTAC}.
\begin{lemma}
\label{transformations}
Given a system matrix $\mathbf{A}$, and a set of $N$ sensor observation matrices $\mathbf{C}_1, \mathbf{C}_2, \ldots, \mathbf{C}_{N}$, define $\mathbf{C} \triangleq {\begin{bmatrix}\mathbf{C}^T_{1} & \cdots & \mathbf{C}^T_{N}\end{bmatrix}}^T$. Suppose $(\mathbf{A},\mathbf{C})$ is observable. Then, there exists a similarity transformation matrix $\mathbf{T}$ that transforms the pair $\mathbf{(A,C)}$ to $(\bar{\mathbf{A}},\bar{\mathbf{C}})$, such that
\begin{equation}
\begin{aligned}
\bar{\mathbf{A}} &= \left[
\begin{array}{c|c|cc}
\mathbf{A}_{11}  & \multicolumn{3}{c}{\mathbf{0}} \\
\hline
\mathbf{A}_{21}
 & 
\mathbf{A}_{22} & \multicolumn{2}{c}{\mathbf{0}} \\
\cline{2-4}
\vdots & \vdots & \hspace{-5mm} \ddots  & \vdots \\
\mathbf{A}_{N1}& \mathbf{A}_{N2} \hspace{2mm}\cdots & \mathbf{A}_{N(N-1)} & \multicolumn{1}{|c}{\mathbf{A}_{NN}} 
\end{array}
\right], \\~\\
\bar{\mathbf{C}} &= \begin{bmatrix} \bar{\mathbf{C}}_1 \\ \bar{\mathbf{C}}_2 \vspace{-1mm} \\  \vdots \\ \bar{\mathbf{C}}_N \end{bmatrix} = \left[ \begin{array}{cccc} \mathbf{C}_{{11}} & \multicolumn{3}{|c}{\mathbf{0}}\\
\hline
\mathbf{C}_{{21}} & \multicolumn{1}{c}{\mathbf{C}_{{22}}} & \multicolumn{2}{|c}{\mathbf{0}}\\
\hline
\vdots&\vdots&\vdots&\vdots\\
\mathbf{C}_{{N1}} & \multicolumn{1}{c}{\mathbf{C}_{{N2}}}  & \cdots  \mathbf{C}_{N(N-1)} & \mathbf{C}_{NN}\\
 \end{array}
 \right],
\end{aligned}
\label{eqn:gen_form}
\end{equation}
and the pair $(\mathbf{A}_{ii},\mathbf{C}_{ii})$ is observable $\forall i \in \{1,2, \ldots, N\}$.
\end{lemma}

We use the matrix $\mathbf{T}$ given by Lemma \ref{transformations} to perform the coordinate transformation $\mathbf{x}[k]=\mathbf{T}\mathbf{z}[k]$,  yielding:
\begin{equation}
\begin{aligned}
\mathbf{z}[k+1]&=\bar{\mathbf{A}}\mathbf{z}[k],\\
\mathbf{y}_i[k]&=\bar{\mathbf{C}}_i\mathbf{z}[k], \quad \forall i \in \{1, \ldots, N\},
\end{aligned}
\label{eqn:coordinatetransform}
\end{equation}
where $\bar{\mathbf{A}}={\mathbf{T}}^{-1}\mathbf{A}\mathbf{T}$ and $\bar{\mathbf{C}}_i = \mathbf{C}_i\mathbf{T}$ are given by (\ref{eqn:gen_form}). Commensurate with the structure of $\bar{\mathbf{A}}$, the vector $\mathbf{z}[k]$ is of the following form:
\begin{equation}
\mathbf{z}[k]={\begin{bmatrix}
{\mathbf{z}^{(1)}_{}[k]}^{T}&
\cdots&
{\mathbf{z}^{(N)}_{}[k]}^{T}
\end{bmatrix}}^{T},
\label{eqn:substates}
\end{equation}
where $\mathbf{z}^{(j)}[k]$ will be referred to as the $j$-th substate. By construction, since the pair $(\mathbf{A}_{jj},\mathbf{C}_{jj})$ is locally observable w.r.t. the measurements of node $j$, node $j$ will be viewed as the unique source node for substate $j$. In this sense, the role of node $j$ will be to ensure that each non-source node $i\in\mathcal{V}\setminus\{j\}$ maintains an asymptotically correct estimate of substate $j$. For a time-invariant strongly-connected graph, this is achieved in \cite{mitraTAC} by first constructing a spanning tree rooted at node $j$, and then requiring nodes to only listen to their parents in such a tree for estimating substate $j$. The unidirectional flow of information (from the source node $j$ to the rest of the network) so achieved guarantees stability of the error process. The above strategy is no longer applicable when the underlying communication graph is time-varying, for the following reasons. (i) For a given substate $j$, there may not exist a common spanning tree rooted at node $j$ in each graph $\mathcal{G}[k], k\in\mathbb{N}$. (ii) Assuming that a specific spanning tree rooted at node $j$ is guaranteed to repeat (not necessarily periodically) after a finite duration of time, is restrictive, and qualifies as only a special case of Assumption \ref{assump:connectivity}. (iii) Suppose for simplicity that $\mathcal{G}[k]$ is strongly-connected at each time-step (an assumption also made in \cite{wang2}), and hence, there exists a spanning tree $\mathcal{T}_j[k]$ rooted at node $j$ in each such graph. For estimating substate $j$, suppose consensus at time-step $k$ is performed along the spanning tree $\mathcal{T}_j[k]$ (assuming that it is possible to construct such trees in the first place at each time-step). As we demonstrate in the next section, switching between such spanning trees can lead to unstable error processes over time. Thus, if one makes no further assumptions on the system model (beyond joint observability), or the sequence of communication graphs (beyond joint strong-connectivity), ensuring stability of the estimation error dynamics becomes a challenging proposition. Nonetheless, in Section \ref{sec:algo}, we develop a fairly simple algorithm for tackling this problem. In the following section, we highlight the intuition behind our approach via an illustrative example. We close this section by commenting on the assumptions made in this paper.
\begin{remark}
\label{rem:assump}
While we work under the assumption that $(\mathbf{A},\mathbf{C})$ is observable, extension to the case when $(\mathbf{A},\mathbf{C})$ is only detectable is trivial (with appropriate implications for convergence rates), and hence not discussed explicitly. Likewise, the assumption of joint strong-connectivity can also be relaxed to only requiring jointly rooted graphs at the source node for each substate.
\end{remark}
\section{Illustrative Example}
\label{sec:example}
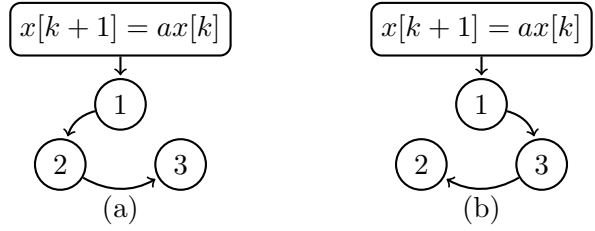
\begin{figure}[t]
\begin{center}
\begin{tikzpicture}
[->,shorten >=1pt,scale=.4, minimum size=5pt, auto=center, node distance=2cm,
  thick,  node/.style={circle, draw=black, thick},]
\tikzstyle{block} = [rectangle, draw, text centered, rounded corners, minimum height=0.7cm, minimum width=0.5cm];
\node [block]  at (-2,2.5) (plant1) {$x[k+1]=ax[k]$};
\node [circle, draw](n1) at (-2,0)  (1)  {1};
\node [circle, draw](n2) at (-4,-2)   (2)  {2};
\node [circle, draw](n3) at (0,-2)   (3)  {3};
\node [ ] ( ) at (-2,-3.5) () {(a)};
\path[every node/.style={font=\sffamily\small}]
(plant1) 
        edge [] node [] {} (1)

    (1)
        edge [bend right] node [] {} (2)
        
    (2)
        edge [bend right] node [] {} (3);
        
\node [block]  at (10,2.5) (plant2) {$x[k+1]=ax[k]$};   
\node [circle, draw](n1) at (10,0)     (4)  {1};
\node [circle, draw](n2) at (8,-2)   (5)  {2};
\node [circle, draw](n3) at (12,-2)   (6)  {3};
\node [ ] ( ) at (10,-3.5) () {(b)};
\path[every node/.style={font=\sffamily\small}]
(plant2) 
        edge [] node [] {} (4)

(4)
       edge [bend left]  node [] {} (6)
        
 (6)   edge [bend left]  node [] {} (5);
       
\end{tikzpicture}
\end{center}
\caption{An LTI system is monitored by a network of 3 nodes, where the communication graph $\mathcal{G}[k]$ switches between the two graphs shown above.}
\label{fig:example}
\end{figure}
\begin{figure}[t]
\begin{center}
\begin{tabular}{cc}
\includegraphics[height=6cm, width=6cm]{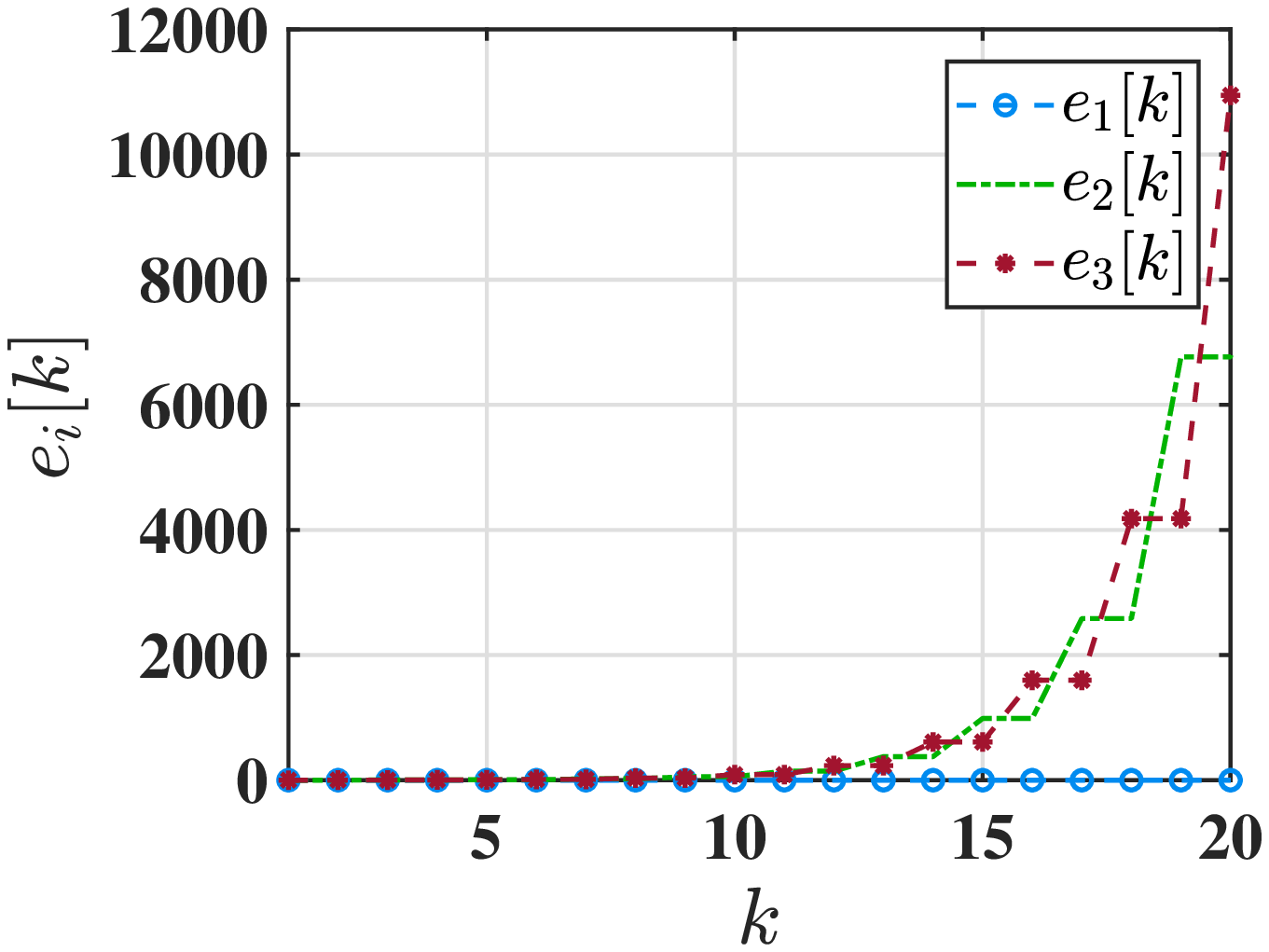}&\includegraphics[height=6cm, width=6cm]{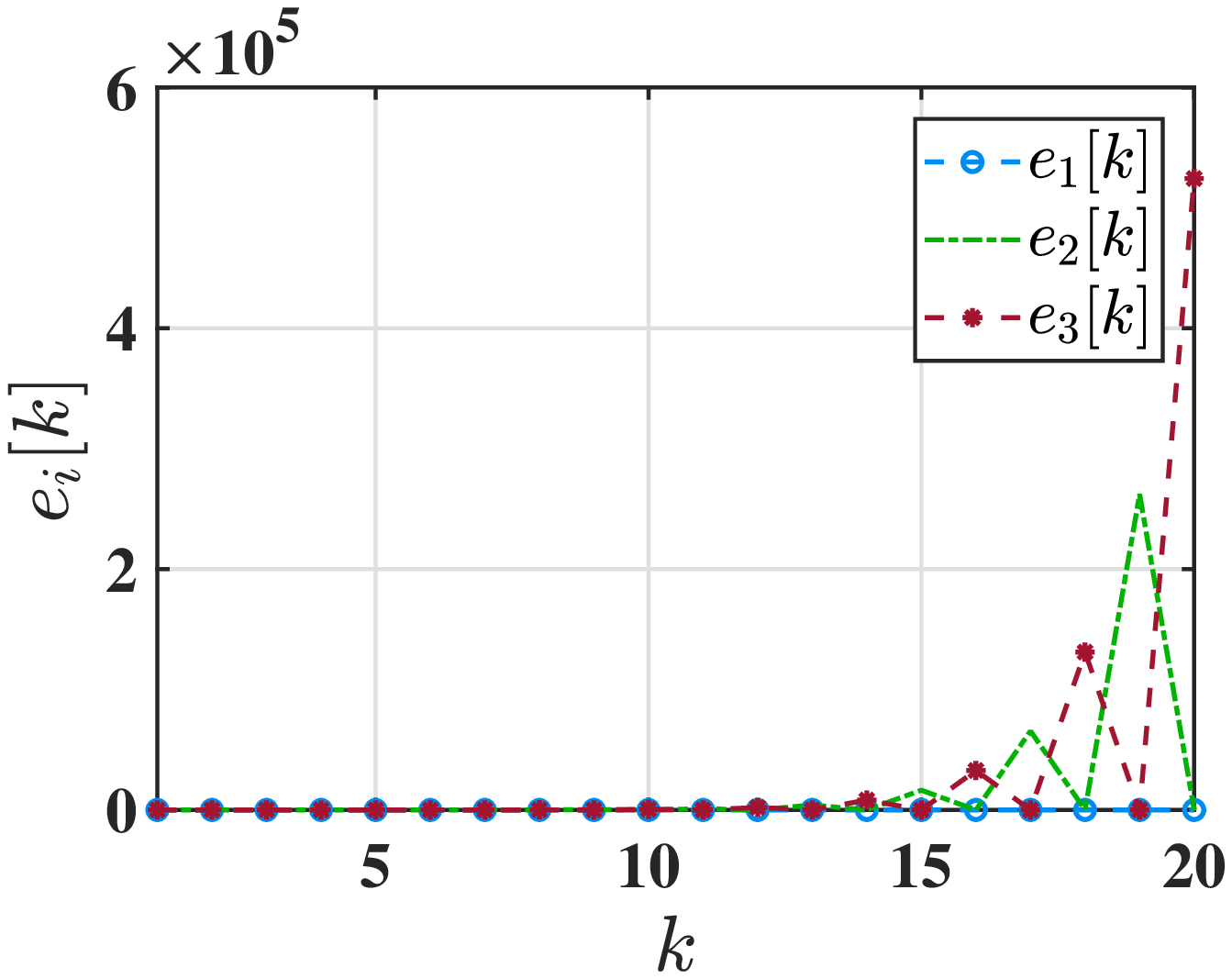}
\end{tabular}
\end{center}
\caption{Estimation error plots of the nodes for the model in Figure \ref{fig:example}. Simulations are performed for a model where $a=2$. The figure on the left corresponds to the case where consensus weights are distributed uniformly among neighbors, while the one on the right is the case where weights are placed along a tree rooted at node 1.}
\label{fig:error}
\end{figure}
Consider a network of 3 nodes monitoring a scalar unstable process $x[k+1]=ax[k]$, as shown in Figure \ref{fig:example}. The communication graph $\mathcal{G}[k]$ switches between the two topologies shown in Figure \ref{fig:example}. Specifically, $\mathcal{G}[k]$ is the graph in Figure $\ref{fig:example}$(a) at all even time-steps, and the one in $\ref{fig:example}$(b) at all odd time-steps. Node 1 is the only node with non-zero measurements, and thus acts as the source node for this network. Suppose for simplicity that it has perfect information of the state at all time-steps, i.e., $\hat{x}_1[k]=x[k], \forall k\in\mathbb{N}$. Given this setup, a standard consensus based state estimate update rule would take the form (see for example \cite{martins,mitraTAC,wang2}):
\begin{equation}
    \hat{x}_i[k+1]=a\left(\sum_{j\in\mathcal{N}_i[k]\cup\{i\}}w_{ij}[k]\hat{x}_j[k]\right), i\in\{2,3\},
    \label{eqn:sampleupdate}
\end{equation}
where the weights $w_{ij}[k]$ are non-negative, and satisfy $\sum_{j\in\mathcal{N}_i[k]\cup\{i\}}w_{ij}[k]=1, \forall k\in\mathbb{N}$. The key question is: how should the consensus weights be chosen to guarantee stability of the estimation errors of nodes 2 and 3? Even for this simple example, if such weights are chosen naively, then the errors may grow unbounded over time. To see this, consider the following two choices: (1) consensus weights are distributed evenly over the set $\mathcal{N}_i[k]\cup\{i\}$, and (2) consensus weights are placed along the tree rooted at node $1$ (i.e., when $\mathcal{G}[k]$ is the graph in Figure \ref{fig:example}(a) (resp., Figure \ref{fig:example}(b)), node 2 (resp., node 3) listens to only node 1, while node 3 (resp., node 2) listens to only node 2 (resp., node 3)). In each case, the error dynamics are unstable, as depicted in Figure \ref{fig:error}. To come up with a solution to this problem, suppose nodes 2 and 3 are aware of the fact that node 1 has perfect information of the state. Since nodes 2 and 3 have no measurements of their own, intuitively, it makes sense that they should place their consensus weights entirely on node 1 whenever possible. The trickier question for node 2 (resp., node 3) is to decide \textit{when} it should listen to node 3 (resp., node 2). Let us consider the situation from the perspective of node 2. At time-step 0, it adopts the information of node 1, and hence, the error of node 2 is zero at time-step 1. However, the error of node 3 is not necessarily zero at time-step 1. Consequently, if node 2 places a non-zero consensus weight on the estimate of node 3 at time-step 1, its error at time-step 2 would assume a non-zero value. Clearly, at time-step 1, node 2 is better off rejecting the information from node 3, and simply running open-loop. The main take-away point here is that adoption or rejection of information from a neighbor should be based on the quality of information that such a neighbor has to offer. In particular, a node that has come in contact with node 1 more recently is expected to have better information about the state (in this case, perfect information) than the other. In other words, to dynamically evaluate the quality of an estimate, the above reasoning suggests the need to introduce a metric that keeps track of how delayed that estimate is w.r.t. the estimate of the source node 1. We formalize the above observations by introducing such a metric in the next section.
\section{Algorithm}
\label{sec:algo}
Building on the intuition developed in the previous section, we introduce a new approach to designing distributed observers for a general class of time-varying networks. The main idea is the use of a ``freshness-index"  that keeps track of how delayed the estimates of a node are w.r.t. the estimates of a source node. Specifically, for updating its estimate of $\mathbf{z}^{(j)}[k]$, each node $i\in\mathcal{V}$ maintains and updates at every time-step a freshness-index $\tau^{(j)}_i[k]$. At each time-step $k\in\mathbb{N}$, the index $\tau^{(j)}_i[k]$ plays the following role: it determines whether node $i$ should adopt the information received from one of its neighbors in $\mathcal{N}_i[k]$, or run open-loop, for updating $\hat{\mathbf{z}}^{(j)}_i[k]$. In case it is the former, it also indicates which specific neighbor in $\mathcal{N}_i[k]$ that node $i$ should listen to at time-step $k$; this piece of information is particularly important for the problem under consideration, and ensures stability of the error process. A formal description of the rules that govern the update of the estimates of the $j$-th substate $\mathbf{z}^{(j)}[k]$ is as follows.
\begin{itemize}
    \item \underline{\textbf{Initialization of Freshness-Indices:}} Each node $i\in\mathcal{V}$ maintains an index $\tau^{(j)}_i[k]\in\{\omega\}\cup\mathbb{N}$, where $\omega$ is a dummy value. Specifically, $\tau^{(j)}_i[k]=\omega$  represents an ``infinite-delay"  w.r.t. the estimate of the source node for sub-state $j$, namely node $j$ (i.e., it represents that node $i$ has not received any information from node $j$ regarding substate $j$ up to time-step $k$). The indices $\tau^{(j)}_i[k]$ are initialized as: $\tau^{(j)}_j[0]=0, \tau^{(j)}_i[0]=\omega, \forall i\in\mathcal{V}\setminus\{j\}$.
    \vspace{3mm}
    \item \underline{\textbf{Update Rules for the Source Node:}} Node $j$ maintains $\tau^{(j)}_j[k]=0, \forall k\in\mathbb{N}$, and updates $\hat{\mathbf{z}}^{(j)}_j[k]$ as:
    \begin{equation}
     \hat{\mathbf{z}}^{(j)}_j[k+1]=(\mathbf{A}_{jj}-\mathbf{L}_j\mathbf{C}_{jj})\hat{\mathbf{z}}^{(j)}_j[k]+\sum \limits_{q=1}^{(j-1)}(\mathbf{A}_{jq}-\mathbf{L}_j\mathbf{C}_{jq})\hat{\mathbf{z}}^{(q)}_j[k]+\mathbf{L}_j\mathbf{y}_j[k],
     \label{eqn:sourceupdate}
    \end{equation}
where $\mathbf{L}_j$ is an output-injection gain to be decided later.
\vspace{3mm}
\item \underline{\textbf{Update Rules for the Non-Source Nodes:}} For each non-source node $i\in\mathcal{V}\setminus\{j\}$, we consider two distinct cases based on the value of $\tau^{(j)}_i[k]$.

\underline{\textbf{Case 1: $\tau^{(j)}_i[k]=\omega$}}. Define 
\begin{equation}
    \mathcal{M}^{(j)}_i[k]\triangleq\{l\in\mathcal{N}_i[k]: \tau^{(j)}_l[k]\neq\omega\}.
\label{eqn:setMdefn}
\end{equation}
If $\mathcal{M}^{(j)}_i[k]\neq\emptyset$, let $u=\argmin_{l\in\mathcal{M}^{(j)}_i[k]} \tau^{(j)}_l[k]$.\footnote{We drop the dependence of $u$ on parameters $i,j$ and $k$ for clarity.} If more than one node satisfies the above criterion, node $i$ picks any such one arbitrarily. It then updates $\tau^{(j)}_i[k]$ and $\hat{\mathbf{z}}^{(j)}_i[k]$ as follows:
\begin{equation}
    \tau^{(j)}_i[k+1]=\tau^{(j)}_u[k]+1,
    \label{eqn:indexupdatecase11}
\end{equation}
\begin{equation}
    \hat{\mathbf{z}}^{(j)}_i[k+1]=\mathbf{A}_{jj}\hat{\mathbf{z}}^{(j)}_u[k]+\sum \limits_{q=1}^{(j-1)}\mathbf{A}_{jq}\hat{\mathbf{z}}^{(q)}_i[k].
    \label{eqn:nonsourcecase1}
\end{equation}
If $\mathcal{M}^{(j)}_i[k]=\emptyset$, then
\begin{equation}
    \tau^{(j)}_i[k+1]=\omega,
    \label{eqn:indexupdatecase12}
\end{equation}
\begin{equation}
    \hat{\mathbf{z}}^{(j)}_i[k+1]=\mathbf{A}_{jj}\hat{\mathbf{z}}^{(j)}_i[k]+\sum \limits_{q=1}^{(j-1)}\mathbf{A}_{jq}\hat{\mathbf{z}}^{(q)}_i[k].
  \label{eqn:nonsourcecase2}
\end{equation}

 \underline{\textbf{Case 2: $\tau^{(j)}_i[k]\neq\omega$}}. Define 
\begin{equation}
    \mathcal{F}^{(j)}_i[k]\triangleq\{l\in\mathcal{M}^{(j)}_i[k]: \tau^{(j)}_l[k] < \tau^{(j)}_i[k]\},
\label{eqn:setFdefn}
\end{equation}
where $\mathcal{M}^{(j)}_i[k]$ is as defined in \eqref{eqn:setMdefn}. If $\mathcal{F}^{(j)}_i[k]\neq\emptyset$, let $u=\argmin_{l\in\mathcal{F}^{(j)}_i[k]} \tau^{(j)}_l[k]$. If more than one node satisfies the above criterion, node $i$ picks any such one arbitrarily. It then updates $\tau^{(j)}_i[k]$ as per \eqref{eqn:indexupdatecase11}, and  $\hat{\mathbf{z}}^{(j)}_i[k]$ as per \eqref{eqn:nonsourcecase1}. If $\mathcal{F}^{(j)}_i[k]=\emptyset$, then $\tau^{(j)}_i[k]$ is updated as
\begin{equation}
    \tau^{(j)}_i[k+1]=\tau^{(j)}_i[k]+1,
    \label{eqn:indexupdatecase2}
\end{equation}
and $\hat{\mathbf{z}}^{(j)}_i[k]$ is updated as per \eqref{eqn:nonsourcecase2}.
\end{itemize}

The above steps describe an approach for estimating $\mathbf{z}[k]$, and hence $\mathbf{x}[k]$, since $\mathbf{x}[k]=\mathbf{Tz}[k]$. We now briefly describe each rule of the proposed algorithm. Consider any substate $j\in\{1,\ldots,N\}.$ For estimation of substate $j$, since delays are measured w.r.t. the source node $j$, node $j$ maintains its freshness-index $\tau^{(j)}_j[k]$ at zero for all time, to indicate a zero delay w.r.t. itself. For updating its estimate of $\mathbf{z}^{(j)}[k]$, it uses only its own information, as is evident from \eqref{eqn:sourceupdate}. Every other node starts out with an ``infinite-delay" w.r.t. the source, which is represented by the freshness-index taking on the value $\omega$. The freshness-index of a node $i\in\mathcal{V}\setminus\{j\}$ changes from $\omega$ to a finite value when it comes in contact with a neighbor with a finite delay, i.e., with a freshness-index that is not $\omega$ (see equation \eqref{eqn:setMdefn}). At this point, we say that $\tau^{(j)}_i[k]$ has been ``triggered". Once triggered, at each time-step $k$, a non-source node $i$ will adopt the information of a neighbor $l \in \mathcal{N}_i[k]$ only if node $l$'s estimate of $\mathbf{z}^{(j)}[k]$ is ``more fresh" relative to its own, i.e., only if $\tau^{(j)}_l[k] < \tau^{(j)}_i[k]$ (see equation \eqref{eqn:setFdefn}). Among the set of neighbors in $\mathcal{M}^{(j)}_i[k]$ (if $\tau^{(j)}_i[k]$ has not yet been triggered), or in $\mathcal{F}^{(j)}_i[k]$ (if $\tau^{(j)}_i[k]$ has been triggered), node $i$ only adopts the information (based on \eqref{eqn:nonsourcecase1}) of the neighbor $u$ with the least delay. At this point, the delay of node $i$ matches that of node $u$, and this fact is captured by the update rule \eqref{eqn:indexupdatecase11}. In case node $i$ has no informative neighbor, it increments its own freshness-index linearly by $1$ (to capture the effect of its own information getting older) via the update rule \eqref{eqn:indexupdatecase2}, and runs open-loop based on \eqref{eqn:nonsourcecase2}.\footnote{When $\tau^{(j)}_i[k]=\omega$, any neighbor of node $i$ with a finite delay is informative, while if $\tau^{(j)}_i[k]\neq\omega$, only neighbors with strictly lower freshness-indices are considered informative by node $i$.} Based on the above rules, at any given time-step $k$, $\tau^{(j)}_i[k]$ measures the age-of-information of $\hat{\mathbf{z}}^{(j)}_i[k]$, relative to the source node $j$. This fact is shown later in Lemma \ref{lemma:form}.  
\section{Main Result and Analysis}
\label{sec:mainresult}
The main result of the paper is as follows.
\begin{theorem}
Given an LTI system  \eqref{eqn:system}, and a measurement model \eqref{eqn:Obsmodel}, suppose  $(\mathbf{A,C})$ is observable. Let the sequence of communication graphs $\{\mathcal{G}[k]\}_{k=0}^{\infty}$ satisfy Assumption \ref{assump:connectivity}. Then, the observer gains $\mathbf{L}_1, \ldots, \mathbf{L}_{N}$ can be designed in a manner such that the estimation error of each node $i\in\mathcal{V}$ converges to zero exponentially fast at any desired convergence rate $\rho$, based on the proposed algorithm.
\label{thm:main}
\end{theorem}

In the remainder of this section, we develop a proof of the above result. The idea behind the proof is simple, and as follows. For each substate $j\in\{1,\ldots,N\}$, we want to establish the following three facts: (i) under Assumption \ref{assump:connectivity}, the freshness-index $\tau^{(j)}_i[k]$ of each node $i\in\mathcal{V}$ is guaranteed to get triggered after a finite period of time, (ii) the error in estimation (for substate $j$) of each non-source node $i\in\mathcal{V}\setminus\{j\}$ can be expressed as a delayed version of the estimation error of the source node $j$, (iii) under Assumption \ref{assump:connectivity}, such a delay is bounded above by a constant that depends only on the number of nodes $N$, and the parameter $T$ in Assumption \ref{assump:connectivity}. Based on the above facts, every non-source node $i\in\mathcal{V}\setminus\{j\}$ will inherit the same exponential convergence to the true dynamics $\mathbf{z}^{(j)}[k]$ as that achieved by the corresponding source node $j$. We begin with a simple result (Lemma \ref{lemma:sourceuse}) which states that a non-source node for a certain substate will always listen to the source node for the corresponding substate, whenever it is in a position to do so. Facts (i), (ii), and (iii) (as described above) are then established in Lemmas \ref{lemma:counterinit}, \ref{lemma:form}, and \ref{lemma:delaybound}, respectively.
\begin{lemma}
Consider any substate $j$, and suppose that at some time-step $k$, we have that  $j\in\mathcal{M}^{(j)}_i[k]$, for some $i\in\mathcal{V}\setminus\{j\}$. Then, based on the rules of the proposed algorithm, the following are true.
\begin{enumerate}
    \item[(i)] If $\tau^{(j)}_i[k]=\omega$, then $j=\argmin_{l\in\mathcal{M}^{(j)}_i[k]} \tau^{(j)}_l[k]$.
    \item[(ii)] If $\tau^{(j)}_i[k]\neq\omega$, then $j\in\mathcal{F}^{(j)}_i[k]$, and $j=\argmin_{l\in\mathcal{F}^{(j)}_i[k]} \tau^{(j)}_l[k]$.
\end{enumerate}
\label{lemma:sourceuse}
\end{lemma}
\begin{proof}
The result follows from two simple observations that are direct consequences of the rules of the proposed algorithm: (i) $\tau^{(j)}_j[k]=0, \forall k\in\mathbb{N}$, and (ii) for any $i\in\mathcal{V}\setminus\{j\}$, $\tau^{(j)}_i[k]\geq 1$ whenever $\tau^{(j)}_i[k]\neq\omega$. In other words, the source node for a given substate has the lowest freshness-index for that substate.
\end{proof}
\begin{lemma}
Suppose Assumption \ref{assump:connectivity} is met. Then, for each substate $j$, the following is true based on the rules of the proposed algorithm.
\begin{equation}
    \tau^{(j)}_i[k]\neq\omega,\forall k\geq (N-1)T, \forall i\in\mathcal{V}.
    \label{eqn:counterinit}
\end{equation}
\label{lemma:counterinit}
\end{lemma}
\begin{proof}
Fix a substate $j$, and notice that the claim is true for the corresponding source node $j$, since $\tau^{(j)}_j[k]=0,\forall k\in \mathbb{N}$. Let $\mathcal{C}^{(j)}_0=\{j\}$, and define:
\begin{equation}
    \mathcal{C}^{(j)}_1\triangleq\{i\in\mathcal{V}\setminus\mathcal{C}^{(j)}_0: \{\bigcup \limits_{\tau=0}^{T-1}\mathcal{N}_i[\tau]\}\cap\mathcal{C}^{(j)}_0\neq\emptyset\}.
\end{equation}
In words, $\mathcal{C}^{(j)}_1$ represents the set of non-source nodes that have a direct edge from node $j$ at least once over the interval $[0,T)$. Based on Assumption \ref{assump:connectivity}, $\mathcal{C}^{(j)}_1$ is non-empty (barring the trivial case when $\mathcal{V}=\{j\}$). For each $i\in\mathcal{C}^{(j)}_1$, it must be that $j \in \mathcal{M}^{(j)}_i[k]$ for some $k\in[0,T)$. Thus, based on \eqref{eqn:indexupdatecase11} and \eqref{eqn:indexupdatecase2}, it must be that $\tau^{(j)}_i[k]\neq\omega, \forall k\geq T, \forall i\in \mathcal{C}^{(j)}_1$. We can keep repeating the above argument by recursively defining the sets $\mathcal{C}^{(j)}_r,1\leq r \leq (N-1)$, as follows.
\begin{equation}
    \mathcal{C}^{(j)}_r\triangleq\{i\in\mathcal{V}\setminus\bigcup \limits_{q=0}^{(r-1)}\mathcal{C}^{(j)}_q:\{\hspace{-2.5mm}\bigcup \limits_{\tau=(r-1)T}^{rT-1}\hspace{-3mm}\mathcal{N}_i[\tau]\}\cap\{\bigcup \limits_{q=0}^{(r-1)}\mathcal{C}^{(j)}_q\}\neq \emptyset\}.
\end{equation}
We proceed via induction on $r$. Suppose it holds that $\tau^{(j)}_i[k]\neq\omega,\forall i\in\bigcup \limits_{q=0}^{(r-1)}\mathcal{C}^{(j)}_q,\forall k\geq(r-1)T$. If $\mathcal{V}\setminus\bigcup \limits_{q=0}^{(r-1)}\mathcal{C}^{(j)}_q$ is empty, then we are done. Else, based on Assumption \ref{assump:connectivity}, it must be that $\mathcal{C}^{(j)}_r$ is non-empty. Based on the induction hypothesis, it also follows that for each $i\in\mathcal{C}^{(j)}_r$, $\mathcal{M}^{(j)}_i[k]\cap\bigcup \limits_{q=0}^{(r-1)}\mathcal{C}^{(j)}_q\neq\emptyset$, for some $k\in[(r-1)T,rT)$. Consequently, based on \eqref{eqn:indexupdatecase11} and \eqref{eqn:indexupdatecase2}, we must have $\tau^{(j)}_i[k]\neq\omega, \forall i\in \mathcal{C}^{(j)}_r, \forall k\geq rT$. Finally, note that repeating the above argument at most $(N-1)$ times exhausts the node set $\mathcal{V}$.
\end{proof}

\begin{lemma}
Consider any substate $j$, and suppose that at some time-step $k$, we have  $\tau^{(j)}_i[k]=m$, where $i\in\mathcal{V}\setminus\{j\}$, and $m\in\mathbb{N}_{+}$. Then, the following is true based on the rules of the proposed algorithm.
\begin{equation}
    \hat{\mathbf{z}}^{(j)}_i[k]=\mathbf{A}_{jj}^{m}\hat{\mathbf{z}}^{(j)}_j[k-m]+\sum_{q=1}^{(j-1)}\hspace{-1mm}\sum_{\tau=(k-m)}^{(k-1)}\hspace{-2.5mm}\mathbf{A}_{jj}^{(k-\tau-1)}\mathbf{A}_{jq}\hat{\mathbf{z}}^{(q)}_{v(\tau)}[\tau],
    \label{eqn:delayedform}
\end{equation}
where $v(\tau)\in\mathcal{V}\setminus\{j\}.$
\label{lemma:form}
\end{lemma}
\begin{proof}
Fix any substate $j$. We prove the result by inducting on $m$. Consider the base case when $m=1$, and suppose that at some time-step $k$, $\tau^{(j)}_i[k]=1$ for some $i\in\mathcal{V}\setminus\{j\}.$ Based on the rules of the proposed algorithm, this is possible if and only if node $j$ is a neighbor of node $i$ at time-step $k-1$, i.e., if and only if $j\in\mathcal{M}^{(j)}_i[k-1]$. Based on Lemma \ref{lemma:sourceuse} and \eqref{eqn:nonsourcecase1}, we must then have: 
\begin{equation}
 \hat{\mathbf{z}}^{(j)}_i[k]=\mathbf{A}_{jj}\hat{\mathbf{z}}^{(j)}_j[k-1]+\sum \limits_{q=1}^{(j-1)}\mathbf{A}_{jq}\hat{\mathbf{z}}^{(q)}_i[k-1].
\end{equation}
Notice that the above equation is of the form \eqref{eqn:delayedform}, with $m=1$ and $v(k-1)=i$. Now suppose that at any given time-step $k$, and for any node $i\in\mathcal{V}\setminus\{j\}$, if $\tau^{(j)}_i[k]\in\{1,2,\ldots,m-1\}$, where $m\geq2$, then the desired identity \eqref{eqn:delayedform} holds. Let at some time-step $k$, $\tau^{(j)}_i[k]=m$ for some $i\in\mathcal{V}\setminus\{j\}.$ Then, one of the two following events occurred at time-step $k-1$: (i) node $i$ adopted the information of some node $u\in\mathcal{N}_i[k-1]$ with $\tau^{(j)}_u[k-1]=(m-1)$, and updated $\hat{\mathbf{z}}^{(j)}_i[k-1]$ based on \eqref{eqn:nonsourcecase1}, and $\tau^{(j)}_i[k-1]$ based on \eqref{eqn:indexupdatecase11}, or  (ii) node $i$ updated $\hat{\mathbf{z}}^{(j)}_i[k-1]$ in an open-loop manner based on \eqref{eqn:nonsourcecase2}, and updated $\tau^{(j)}_i[k-1]$ based on \eqref{eqn:indexupdatecase2}, i.e., $\tau^{(j)}_i[k-1]$ was $(m-1)$. Consider the former scenario, since an identical argument applies to the latter. Since $m\geq2$, $\tau^{(j)}_u[k-1]=(m-1)\geq1$, and hence $u\in\mathcal{V}\setminus\{j\}$. The induction hypothesis thus applies to node $u$. In particular, the following is true:
   \begin{equation}
   \begin{aligned}
    \hat{\mathbf{z}}^{(j)}_u[k-1]&=\mathbf{A}_{jj}^{(m-1)}\hat{\mathbf{z}}^{(j)}_j[(k-1)-(m-1)]\\
    &+\sum_{q=1}^{(j-1)}\hspace{-1mm}\sum_{\tau=((k-1)-(m-1))}^{(k-2)}\hspace{-5mm}\mathbf{A}_{jj}^{((k-1)-\tau-1))}\mathbf{A}_{jq}\hat{\mathbf{z}}^{(q)}_{v(\tau)}[\tau].
    \end{aligned}
    \label{eqn:inducthyp}
\end{equation} 
Combining the above with \eqref{eqn:nonsourcecase1} yields the desired identity \eqref{eqn:delayedform}, with $v(k-1)=i$. This completes the proof.
\end{proof}
\begin{lemma}
Suppose Assumption \ref{assump:connectivity} is met. Then, for any substate $j$, the following is true based on the rules of the proposed algorithm.
\begin{equation}
   \tau^{(j)}_i[k] \leq 2(N-1)T, \forall k\geq (N-1)T, \forall i\in\mathcal{V}.
   \label{eqn:delaybound}
\end{equation}
\label{lemma:delaybound}
\end{lemma}
\begin{proof}
Fix a substate $j$, and observe that the claim is trivially true for the source node $j$, since $\tau^{(j)}_j[k]=0, \forall k\in\mathbb{N}$. Our goal is to analyze how the freshness-indices $\tau^{(j)}_i[k]$ of the non-source nodes $i\in\mathcal{V}\setminus\{j\}$ evolve over the interval $[(N-1)T,2(N-1)T]$. To this end, notice that based on Lemma \ref{lemma:counterinit}, it must be that $\tau^{(j)}_i[(N-1)T]\neq\omega$, and in particular, $\tau^{(j)}_i[(N-1)T]\leq(N-1)T, \forall i\in\mathcal{V}\setminus\{j\}$. The latter claim is a consequence of the following observation, which in turn follows from the rules that govern equations \eqref{eqn:indexupdatecase11} and \eqref{eqn:indexupdatecase2}:
\begin{equation}
    \tau^{(j)}_i[k+1]\leq\tau^{(j)}_i[k]+1, \hspace{1mm} \textrm{whenever} \hspace{1mm}  \tau^{(j)}_i[k]\neq \omega.
    \label{eqn:indexbound}
\end{equation}
As in the proof of Lemma \ref{lemma:counterinit}, let $\mathcal{C}^{(j)}_0=\{j\}$, and define:
\begin{equation}
    \mathcal{C}^{(j)}_1\triangleq\{i\in\mathcal{V}\setminus\mathcal{C}^{(j)}_0: \{\hspace{-3mm}\bigcup \limits_{\tau=(N-1)T}^{NT-1}\hspace{-3.5mm}\mathcal{N}_i[\tau]\}\cap\mathcal{C}^{(j)}_0\neq\emptyset\}.
\end{equation}
In words, $\mathcal{C}^{(j)}_1$ represents the set of nodes that have a direct edge from node $j$ at least once over the interval $[(N-1)T,NT-1)$. Based on Assumption \ref{assump:connectivity}, unless $\mathcal{V}=\{j\}$ (in which case the claim holds trivially), $\mathcal{C}^{(j)}_1$ is non-empty. For each node $i\in\mathcal{C}^{(j)}_1$, it is easy to verify that based on Lemma \ref{lemma:sourceuse}, the update rules \eqref{eqn:indexupdatecase11}, \eqref{eqn:indexupdatecase2}, and equation \eqref{eqn:indexbound}, that  $\tau^{(j)}_i[NT]\leq T$. We proceed by defining the sets $\mathcal{C}^{(j)}_r,1\leq r \leq (N-1)$ recursively, and then inducting on $r$.
\begin{equation}
    \mathcal{C}^{(j)}_r\triangleq\{i\in\mathcal{V}\setminus\bigcup \limits_{q=0}^{(r-1)}\mathcal{C}^{(j)}_q:\{\hspace{-3.5mm}\bigcup \limits_{\tau=(N+r-2)T}^{(N+r-1)T-1}\hspace{-4.5mm}\mathcal{N}_i[\tau]\}\cap\{\bigcup \limits_{q=0}^{(r-1)}\mathcal{C}^{(j)}_q\}\neq \emptyset\}.
\end{equation}
Suppose the following is true:  $\tau^{(j)}_i[(N-1+q)T]\leq  qT,\forall i \in \mathcal{C}^{(j)}_{q}$, where $1\leq q \leq (r-1)$. Suppose $\mathcal{V}\setminus\bigcup \limits_{q=0}^{(r-1)}\mathcal{C}^{(j)}_q$ is non-empty (for the case when it is empty, the proof can be completed similarly). Based on Assumption \ref{assump:connectivity}, it must be that $\mathcal{C}^{(j)}_r$ is non-empty. Consider a node $i\in\mathcal{C}^{(j)}_r$. Based on the way $\mathcal{C}^{(j)}_r$ is defined, we know that at some time-step $k\in[(N+r-2)T,(N+r-1)T)$, node $i$ has a neighbor $u$ from the set $\bigcup \limits_{q=0}^{(r-1)}\mathcal{C}^{(j)}_q$. Based on the induction hypothesis and \eqref{eqn:indexbound}, it must be that $\tau^{(j)}_u[k] \leq (k-(N-1)T)$. At this point, there are two possibilities: (i) $u\in\mathcal{F}^{(j)}_i[k]$, and node $i$ updates $\tau^{(j)}_i[k]$ based on \eqref{eqn:indexupdatecase11}, or (ii) $u\notin\mathcal{F}^{(j)}_i[k]$ (implying that $\tau^{(j)}_i[k]\leq\tau^{(j)}_u[k]\leq(k-(N-1)T)$), and node $i$ updates $\tau^{(j)}_i[k]$ based on \eqref{eqn:indexupdatecase2}. In either case, it follows from \eqref{eqn:indexbound} that $\tau^{(j)}_i[(N-1+r)T]\leq rT$. Thus, we have established that $\tau^{(j)}_i[2(N-1)T]\leq (N-1)T, \forall i\in\bigcup \limits_{q=0}^{(N-1)}\mathcal{C}^{(j)}_q=\mathcal{V}$, where the fact that $\bigcup \limits_{q=0}^{(N-1)}\mathcal{C}^{(j)}_q=\mathcal{V}$ follows from Assumption \ref{assump:connectivity} and the manner in which the $\mathcal{C}^{(j)}_r$ sets are defined. Notice that in terms of the bound on the freshness-indices, we are essentially back to the same scenario as that at time-step $(N-1)T$. Thus, one can repeat the above argument to establish that $\tau^{(j)}_i[m(N-1)T]\leq(N-1)T, \forall i\in\mathcal{V}, \forall m\in\mathbb{N}_{+}$. Finally, based on the above bound and \eqref{eqn:indexbound}, notice that for any node $i\in\mathcal{V}$, $\tau^{(j)}_i[k]$ is strictly upper-bounded by $2(N-1)T$ at any time-step $k\in(m(N-1)T,(m+1)(N-1)T)$, where $m\in\mathbb{N}_{+}$.
\end{proof}
\begin{proof} \textbf{(Theorem \ref{thm:main})} 
Given a desired rate of convergence $\rho\in(0,1)$, pick a set of positive scalars $\{\rho_1,\ldots,\rho_N\}$, such that $\rho_1 < \rho_2 < \cdots \rho_N < \rho$. For each substate $j\in\{1,\ldots,N\}$, let the corresponding source node $j$ design the observer gain $\mathbf{L}_j$ (featuring in equation \eqref{eqn:sourceupdate}) in a manner such that the matrix $(\mathbf{A}_{jj}-\mathbf{L}_{j}\mathbf{C}_{jj})$ has distinct real eigenvalues with spectral radius equal to $\rho_j$. Such a choice of $\mathbf{L}_j$ exists as the pair $(\mathbf{A}_{jj},\mathbf{C}_{jj})$ is observable by construction. Then, there exists a set of positive scalars $\{\alpha_1, \ldots,\alpha_{N}\}$, such that \cite{horn}:\footnote{We use $\left\Vert\mathbf{A}\right\Vert$ to refer to the induced 2-norm of a matrix $\mathbf{A}$.}
\begin{equation}
    \left\Vert{(\mathbf{A}_{jj}-\mathbf{L}_{j}\mathbf{C}_{jj})}^k\right\Vert \leq \alpha_{j}\rho_{j}^k, \forall k\in\mathbb{N}.
    \label{eqn:schurbound}
\end{equation}
For a particular substate $j$, let $\mathbf{e}^{(j)}_i[k]=\hat{\mathbf{z}}^{(j)}_i[k]-\mathbf{z}^{(j)}[k]$. Consider the first substate $j=1$, and observe that based on \eqref{eqn:gen_form}, \eqref{eqn:coordinatetransform},    and \eqref{eqn:sourceupdate}, the following is true: $\mathbf{e}^{(1)}_{1}[k+1]=(\mathbf{A}_{11}-\mathbf{L}_{1}\mathbf{C}_{11})\mathbf{e}^{(1)}_{1}[k]$. Thus, we obtain
\begin{equation}
    \mathbf{e}^{(1)}_{1}[k]={(\mathbf{A}_{11}-\mathbf{L}_{1}\mathbf{C}_{11})}^k\mathbf{e}^{(1)}_1[0].
    \label{eqn:mode1rolledout}
\end{equation}
Taking norms on both sides of \eqref{eqn:mode1rolledout}, and using \eqref{eqn:schurbound}, yields:
\begin{equation}
    \left\Vert\mathbf{e}^{(1)}_1[k]\right\Vert \leq c_1\rho_{1}^k, \forall k\in\mathbb{N},
    \label{eqn:sourcebound1}
\end{equation}
where $c_1\triangleq\alpha_1\left\Vert\mathbf{e}^{(1)}_1[0]\right\Vert$. Based on Lemmas \ref{lemma:counterinit} and \ref{lemma:form}, and the fact that $\mathbf{z}^{(1)}[k]={(\mathbf{A}_{11})}^m\mathbf{z}^{(1)}[k-m], \forall m\in\mathbb{N}$, the following is true for any non-source node $i\in\mathcal{V}\setminus\{1\}$:
\begin{equation}
   \mathbf{e}^{(1)}_i[k]={(\mathbf{A}_{11})}^{\tau^{(1)}_i[k]}\mathbf{e}^{(1)}_1[k-\tau^{(1)}_i[k]], \forall k\geq(N-1)T.
   \label{eqn:errmode1}
\end{equation}
For each substate $j$, one can always find scalars $\beta_j, \gamma_j \geq 1$, such that $\left\Vert{(\mathbf{A}_{jj})}^{k}\right\Vert\leq\beta_j\gamma^k_j, \forall k\in\mathbb{N}$ \cite{horn}.\footnote{Note that such a bound also applies to the case when $\mathbf{A}_{jj}$ is Schur.} Using this bound and the one in \eqref{eqn:sourcebound1}, the fact that $\gamma_1\geq1$ and $\rho_1<1$, the fact that $\tau^{(1)}_i[k] \leq 2(N-1)T, \forall k \geq (N-1)T$ based on Lemma \ref{lemma:delaybound}, and the sub-multiplicative property of the 2-norm, we obtain the following by taking norms on both sides of \eqref{eqn:errmode1}:
\begin{equation}
    \left\Vert\mathbf{e}^{(1)}_i[k]\right\Vert \leq \bar{c}_1\rho_{1}^k, \forall k\geq(N-1)T, \forall i\in\mathcal{V},
    \label{eqn:boundmode1}
\end{equation}
where 
\begin{equation}
    \bar{c}_1\triangleq c_1\beta_1{\left(\frac{\gamma_1}{\rho_1}\right)}^{2\bar{T}}, \hspace{3mm} \bar{T}=(N-1)T.
    \label{eqn:cbar1}
\end{equation}
Note that $\bar{c}_1 \geq c_1$, and hence the bound in \eqref{eqn:boundmode1} applies to node 1 as well (see equation \eqref{eqn:sourcebound1}). 

Our goal is to now obtain a bound similar to that in \eqref{eqn:boundmode1} for each substate $j\in\{2,\ldots,N\}$. To this end, let  $g_{jq}=\left\Vert(\mathbf{A}_{jq}-\mathbf{L}_q\mathbf{C}_{jq})\right\Vert$, and $h_{jq}=\left\Vert\mathbf{A}_{jq}\right\Vert$. Define the following quantities recursively for $j\in\{2,\ldots,N\}.$
\begin{equation}
    \begin{aligned}
        c_j&\triangleq\frac{\alpha_j}{\rho_{j}^{(2j-3)\bar{T}}}\left(\left\Vert\mathbf{e}^{(j)}_j[(2j-3)\bar{T}]\right\Vert+\sum \limits_{q=1}^{(j-1)}\frac{g_{jq}\bar{c}_q}{(\rho_j-\rho_q)}\rho_{q}^{(2j-3)\bar{T}}\right),\\
        \bar{c}_j&\triangleq\beta_j\left(c_j{\left(\frac{\gamma_j}{\rho_j}\right)}^{2\bar{T}}+\sum\limits_{q=1}^{(j-1)}\frac{h_{jq}\bar{c}_q}{(\gamma_j-\rho_q)}{\left(\frac{\gamma_j}{\rho_q}\right)}^{2\bar{T}}\right),
    \end{aligned}
    \label{eqn:defn}
\end{equation}
where $c_1\triangleq\alpha_1\left\Vert\mathbf{e}^{(1)}_1[0]\right\Vert$, and $\bar{c}_1$ is as defined in \eqref{eqn:cbar1}. 
Based on the above definitions, we claim that for each substate $j\in\{1,\ldots,N\}$, the following is true:
\begin{equation}
    \left\Vert\mathbf{e}^{(j)}_i[k]\right\Vert \leq \bar{c}_j\rho^k_j, \forall k\geq(2j-1)\bar{T}, \forall i \in \mathcal{V}.
    \label{eqn:errboundmodej}
\end{equation}
To prove the above claim, we proceed via induction on the substate number $j$. Suppose the claim holds for all $q\in\{1,\ldots,j-1\}$, where $2\leq j \leq N$. To prove the desired result for substate $j$, observe that equations \eqref{eqn:gen_form} and \eqref{eqn:coordinatetransform} yield:
\begin{equation}
\begin{aligned}
    &\mathbf{z}^{(j)}[k+1]=\mathbf{A}_{jj}\mathbf{z}^{(j)}[k]+\sum_{q=1}^{(j-1)}\mathbf{A}_{jq}\mathbf{z}^{(q)}[k]\\
    \hspace{-3mm}&=(\mathbf{A}_{jj}-\mathbf{L}_j\mathbf{C}_{jj})\mathbf{z}^{(j)}[k]+\sum \limits_{q=1}^{(j-1)}(\mathbf{A}_{jq}-\mathbf{L}_j\mathbf{C}_{jq})\mathbf{z}^{(q)}[k]+\mathbf{L}_j\mathbf{y}_j[k].
\end{aligned}
\label{eqn:modej}
\end{equation}
Based on the above equation and \eqref{eqn:sourceupdate}, we obtain:
\begin{equation}
\mathbf{e}^{(j)}_j[k+1]=(\mathbf{A}_{jj}-\mathbf{L}_j\mathbf{C}_{jj})\mathbf{e}^{(j)}_j[k]+\sum \limits_{q=1}^{(j-1)}(\mathbf{A}_{jq}-\mathbf{L}_j\mathbf{C}_{jq})\mathbf{e}^{(q)}_j[k].
\end{equation}
Rolling out the above equation over time, we obtain:
\begin{equation}
\begin{aligned}
    \mathbf{e}^{(j)}_j[k]&={(\mathbf{A}_{jj}-\mathbf{L}_j\mathbf{C}_{jj})}^{(k-(2(j-1)-1)\bar{T})}\mathbf{e}^{(j)}_j[(2(j-1)-1)\bar{T}]\\
    &+\sum \limits_{q=1}^{(j-1)} \sum  \limits_{\tau=(2(j-1)-1)\bar{T}}^{k-1}\hspace{-4mm}{(\mathbf{A}_{jj}-\mathbf{L}_j\mathbf{C}_{jj})}^{(k-\tau-1)}(\mathbf{A}_{jq}-\mathbf{L}_j\mathbf{C}_{jq})\mathbf{e}^{(q)}_j[\tau],
    \end{aligned}
    \label{eqn:modejerrorrolledout}
\end{equation}
where $k\geq(2(j-1)-1)\bar{T}$. Taking norms on both sides of the above equation, using the sub-multiplicative property of the two-norm, and the triangle inequality, we obtain:
\begin{equation}
    \begin{aligned}
        \left\Vert\mathbf{e}^{(j)}_j[k]\right\Vert &\overset{(a)}{\leq} \alpha_j\rho^k_j\left(\frac{\left\Vert\mathbf{e}^{(j)}_j[(2j-3)\bar{T}]\right\Vert}{\rho_{j}^{(2j-3)\bar{T}}}+\frac{1}{\rho_j}\sum_{q=1}^{(j-1)}g_{jq}\hspace{-4mm}\sum  \limits_{\tau=(2j-3)\bar{T}}^{(k-1)}\hspace{-4mm}\rho^{-\tau}_j\left\Vert\mathbf{e}^{(q)}_j[\tau]\right\Vert\right)\\
        &\overset{(b)}{\leq} \alpha_j\rho^k_j\left(\frac{\left\Vert\mathbf{e}^{(j)}_j[(2j-3)\bar{T}]\right\Vert}{\rho_{j}^{(2j-3)\bar{T}}}+\frac{1}{\rho_j}\sum_{q=1}^{(j-1)}g_{jq} \bar{c}_q\hspace{-3mm}\sum  \limits_{\tau=(2j-3)\bar{T}}^{(k-1)}\hspace{-1mm}{\left(\frac{\rho_q}{\rho_j}\right)}^{\tau}\right)\\
        &\overset{(c)}{\leq} c_j\rho^k_j, \forall k \geq (2j-3)\bar{T}.\\
        \label{eqn:boundsourcemodej}
    \end{aligned}
\end{equation}
In the above inequalities, (a) follows from \eqref{eqn:schurbound} and the definition of $g_{jq}$, (b) follows by noting that $q\leq(j-1)$, $\tau\geq(2(j-1)-1)\bar{T}$, and then applying the induction hypothesis, (c) follows by simplifying the preceding inequality, noting that $\rho_q < \rho_j$ (by design), and using the definition of $c_j$ in \eqref{eqn:defn}. We have thus obtained a bound on the estimation error of substate $j$ for node $j$. To bound the estimation errors of substate $j$ for each non-source node $i\in\mathcal{V}\setminus\{j\}$, note that equation \eqref{eqn:modej} can be rolled out over time to yield:
\begin{equation}
\mathbf{z}^{(j)}[k]=\mathbf{A}_{jj}^{m}\mathbf{z}^{(j)}[k-m]+\sum_{q=1}^{(j-1)}\hspace{-1mm}\sum_{\tau=(k-m)}^{(k-1)}\hspace{-2.5mm}\mathbf{A}_{jj}^{(k-\tau-1)}\mathbf{A}_{jq}\mathbf{z}^{(q)}[\tau].
\end{equation}
Leveraging Lemma \ref{lemma:form}, we can then obtain the following error dynamics for a node $i\in\mathcal{V}\setminus\{j\}, \forall k\geq\bar{T}.$
\begin{equation}
\begin{aligned}
    \mathbf{e}^{(j)}_i[k]&={(\mathbf{A}_{jj})}^{\tau^{(j)}_i[k]}\mathbf{e}^{(j)}_j[k-\tau^{(j)}_i[k]]\\
    &+\sum_{q=1}^{(j-1)}\hspace{-1mm}\sum_{\tau=(k-\tau^{(j)}_i[k])}^{(k-1)}\hspace{-2.5mm}\mathbf{A}_{jj}^{(k-\tau-1)}\mathbf{A}_{jq}\mathbf{e}^{(q)}_{v(\tau)}[\tau].
    \end{aligned}
    \label{eqn:errornonsourcemodej}
\end{equation}
Based on the induction hypothesis and equation \eqref{eqn:boundsourcemodej}, notice that each error term featuring in the RHS of the above equation will decay exponentially provided $k$ is large enough. Specifically, suppose $k\geq(2j-1)\bar{T}$, in which case $k-\tau^{(j)}_i[k]\geq(2j-3)\bar{T}$, since $\tau^{(j)}_i[k]\leq2\bar{T}, \forall k\geq \bar{T}, \forall i\in\mathcal{V}$ based on Lemma \ref{lemma:delaybound}. For $k\geq(2j-1)\bar{T}$, taking norms on both sides of \eqref{eqn:errornonsourcemodej} yields:
\begin{equation}
    \begin{aligned}
    \left\Vert\mathbf{e}^{(j)}_i[k]\right\Vert &\overset{(a)}{\leq} \beta_j\left(c_j{\left(\frac{\gamma_j}{\rho_j}\right)}^{2\bar{T}}\rho^{k}_j+\gamma^{(k-1)}_j\sum_{q=1}^{(j-1)}h_{jq}\bar{c}_q\hspace{-4.5mm}\sum  \limits_{\tau=(k-\tau^{(j)}_i[k])}^{(k-1)}{\left(\frac{\rho_q}{\gamma_j}\right)}^{\tau}\right)\\
         &\overset{(b)}{\leq} \beta_j\left(c_j{\left(\frac{\gamma_j}{\rho_j}\right)}^{2\bar{T}}\rho^{k}_j+\gamma^{(k-1)}_j\sum_{q=1}^{(j-1)}h_{jq}\bar{c}_q\hspace{-1mm}\sum  \limits_{\tau=(k-2\bar{T})}^{(k-1)}{\left(\frac{\rho_q}{\gamma_j}\right)}^{\tau}\right)\\
         &\overset{(c)}{\leq} \beta_j\left(c_j{\left(\frac{\gamma_j}{\rho_j}\right)}^{2\bar{T}}\rho^{k}_j+\sum_{q=1}^{(j-1)}\frac{h_{jq}\bar{c}_q}{(\gamma_j-\rho_q)}{\left(\frac{\gamma_j}{\rho_q}\right)}^{2\bar{T}}\rho^k_q\right)\\
         &\overset{(d)}{\leq}\bar{c}_j\rho^k_j, \forall k\geq(2j-1)\bar{T}.
         \end{aligned}
        \label{eqn:boundnonsourcemodej}
\end{equation}
In the above inequalities, (a) follows from the induction hypothesis, equation \eqref{eqn:boundsourcemodej}, the bounds on the growth and the norm of $\mathbf{A}_{jj}$, and by noting that $\tau^{(j)}_i[k]\leq2\bar{T}$ (based on Lemma \ref{lemma:delaybound}), $\rho_j < 1, \gamma_j\geq1$, (b) follows by suitably changing the lower limit of the inner summation (over time), a change that is warranted since each term in the summation is non-negative, (c) follows by simplifying the preceding inequality, and (d) follows by noting that $\rho_q < \rho_j,  \forall q\in\{1,\ldots,j-1\}$ by design (and hence $\rho^k_q < \rho^k_j, \forall k\in\mathbb{N}$), and using the definition of $\bar{c}_j$ in \eqref{eqn:defn}. Note that the bound obtained in \eqref{eqn:boundnonsourcemodej} for each non-source node $i\in\mathcal{V}\setminus\{j\}$ applies also to the source node $j$, since $\bar{c}_j \geq c_j$. Let $\mathbf{e}_i[k]=\hat{\mathbf{z}}_i[k]-\mathbf{z}[k]$. For any node $i\in\mathcal{V}$, we then obtain the desired result as follows:
\begin{equation}
        \left\Vert\mathbf{e}_i[k]\right\Vert =\sqrt{\sum \limits_{j=1}^{N}{\left\Vert\mathbf{e}^{(j)}_i[k]\right\Vert}^2}
      \leq\left(\sqrt{\sum \limits_{j=1}^{N}\bar{c}^2_j}\right)\rho^{k}, \forall k \geq (2N-1)(N-1)T.
\end{equation}
\end{proof}
\begin{corollary} (\textbf{Finite-Time Convergence})
Suppose the conditions stated in Theorem \ref{thm:main} are met. Then, the observer gains $\mathbf{L}_1, \ldots, \mathbf{L}_{N}$ can be designed in a manner such that the estimation error of each node $i\in\mathcal{V}$ converges to zero in at most $n+2N(N-1)T$ time-steps.
\label{corr:finitetime}
\end{corollary}
\begin{proof}
For each substate $j\in\{1,\ldots,N\}$, let the corresponding source node $j$ design the observer gain $\mathbf{L}_j$ (featuring in equation \eqref{eqn:sourceupdate}) in a manner such that the matrix $(\mathbf{A}_{jj}-\mathbf{L}_{j}\mathbf{C}_{jj})$ has all its eigenvalues at $0$. Such a choice of $\mathbf{L}_j$ exists based on the fact that the pair $(\mathbf{A}_{jj},\mathbf{C}_{jj})$ is observable by construction. Let $n_j=dim(\mathbf{A}_{jj})$. By construction, $(\mathbf{A}_{jj}-\mathbf{L}_j\mathbf{C}_{jj})$ is then a nilpotent matrix of index at most $n_j$. Thus, it is easy to see that $\mathbf{e}^{(1)}_1[k]=\mathbf{0}, \forall k\geq n_1$, based on \eqref{eqn:mode1rolledout}. Referring to equation \eqref{eqn:errmode1}, and noting that $\tau^{(1)}_i[k]\leq2\bar{T}, \forall k \geq \bar{T}, \forall i\in\mathcal{V}$, we obtain: $\mathbf{e}^{(1)}_i[k]=\mathbf{0},\forall k \geq n_1+2\bar{T}, \forall i\in\mathcal{V}$. One can easily generalize this argument to the remaining substates by using an inductive reasoning akin to that employed in the proof of Theorem \ref{thm:main}. In particular, for any substate $j\in\{2,\ldots,N\}$, one can roll out the error dynamics for node $j$ as in \eqref{eqn:modejerrorrolledout}, with $\sum \limits_{q=1}^{(j-1)}n_q+2(j-1)\bar{T}$ as the initial time. By this time, the estimation errors of all nodes on all substates $q\in\{1,\ldots,j-1\}$ would have converged to zero. The nilpotentcy of $(\mathbf{A}_{jj}-\mathbf{L}_j\mathbf{C}_{jj})$ would then imply that $\mathbf{e}^{(j)}_j[k]=\mathbf{0}, \forall k \geq \sum \limits_{q=1}^{j}n_q+2(j-1)\bar{T}$. Based on \eqref{eqn:errornonsourcemodej}, the zero error of the source node $j$ will manifest into zero errors for the non-source nodes with delay at most $2\bar{T}$. Finally, noting that $\sum \limits_{q=1}^{N}n_q=n$ completes the proof.
\end{proof}

\begin{remark} Notice that given a desired convergence rate $\rho$, the general design approach described in the proof of Theorem \ref{thm:main} offers a considerable degree of freedom in choosing the parameters $\rho_1,\ldots,\rho_N$, since they only need to satisfy $0 < \rho_1 < \rho_2 < \cdots < \rho_N < \rho$. As such, this can be achieved in infinitely many ways, and the design flexibility so obtained in choosing the observer gains can be exploited to optimize transient performance, or performance against noise. In contrast, the proof of Corollary \ref{corr:finitetime} highlights a specific approach to obtain finite-time convergence. However, such an approach may lead to undesirable transient spikes in the estimation errors, owing to large observer gains.
\end{remark}
\section{Conclusion}
In this paper, we developed a new approach towards designing distributed observers that work under the basic assumption of joint observability, and can handle a very general class of time-varying graphs. Unlike existing literature, this is achieved without requiring multiple consensus iterations between consecutive time-steps of the dynamics. Instead, our main idea is based on introducing a metric that keeps track of the age-of-information being diffused across the network, and in turn, acts as measure of quality of such information.  We established that any desired exponential convergence rate can be achieved based on our approach. Furthermore, we showed that one can even obtain finite-time convergence via an appropriate choice of the observer gains.
\bibliographystyle{unsrt}
\bibliography{refs}
\end{document}